\documentclass[a4paper,11pt]{article}

\usepackage{amsmath}
\usepackage{amsthm}

\usepackage{a4wide}

\newtheorem{theorem}{Theorem}[section]
\newtheorem{lemma}[theorem]{Lemma}
\newtheorem{corollary}[theorem]{Corollary}
\newtheorem{example}[theorem]{Example}
\newtheorem*{remark}{Remark}

\newcommand{\ssp}{{\sc{Shortest-Superstring}}}
\newcommand{\tsp}{{\sc{Max-ATSP-Path}}}
\newcommand{\ov}{\textrm{ov}}
\newcommand{\pr}{\textrm{pref}}
\newcommand{\imin}{i_{\min}}
\newcommand{\imax}{i_{\max}}
\newcommand{\wmin}{w_{\min}}
\newcommand{\wmax}{w_{\max}}
\newcommand{\pmin}{p_{\min}}
\newcommand{\pmax}{p_{\max}}
\newcommand{\lmin}{l_{\min}}
\newcommand{\lmax}{l_{\max}}
\newcommand{\hwmin}{\pmin(w)}
\newcommand{\hwmax}{\pmax(w)}
\newcommand{\case}[1]{\noindent{\bf Case #1: }}

\title{Lyndon Words and Short Superstrings}
\author{Marcin Mucha\thanks{This work was partially supported by the ERC StG project PAAl no.\ 259515} \\ \texttt{mucha@mimuw.edu.pl}}

\begin{document}

\maketitle

\begin{abstract}
In the shortest superstring problem, we are given a set of strings $\{s_1,\ldots,s_k\}$ and want to find a string that contains all $s_i$ as substrings and has minimum length. This is a classical problem in approximation and the best known approximation factor is $2\frac{1}{2}$, given by Sweedyk~\cite{sweedyk} in 1999. Since then no improvement has been made, howerever two other approaches yielding a $2\frac{1}{2}$-approximation algorithms have been proposed by Kaplan et al.~\cite{kaplan} and recently by Paluch et al.~\cite{paluch} --- both based on a reduction to maximum asymmetric TSP path (\tsp) and structural results of Breslauer et al.~\cite{breslauer}. 

In this paper we give an algorithm that achieves an approximation ratio of $2 \frac{11}{23}$, breaking through the long-standing bound of $2\frac{1}{2}$. We use the standard reduction of \ssp\ to \tsp. The new, somewhat surprising, algorithmic idea is to take the better of the two solutions obtained by using: (a) the currently best $\frac{2}{3}$-approximation algorithm for \tsp\, and (b) a na\"{\i}ve cycle-cover based $\frac{1}{2}$-approximation algorithm. To prove that this indeed results in an improvement, we further develop a theory of string overlaps, extending the results of Breslauer et al.~\cite{breslauer}. This theory is based on the novel use of Lyndon words, as a substitute for generic unbordered rotations and critical factorizations, as used by Breslauer et al. 
\end{abstract}

\section{Introduction}

\paragraph{The Shortest Superstring Problem}
In the \ssp\ problem we are given a set of strings $\{s_1,\ldots,s_k\}$ and want to find a string that contains all $s_i$ as substrings and has minimum length. The problem has several applications including data compression~\cite{gallant,storer} and DNA sequencing~\cite{lesk,li,peltola,waterman}. In the latter, one attempts to reconstruct a DNA molecule, which is a string over the alphabet $\{A,C,G,T\}$, based on a massive set of short fragments. These fragments (i.e.\ substrings) of the molecule can be obtained by \emph{sequencing}. The reconstruction problem can be viewed as a shortest superstring problem based on the premise that the original molecule is a superstring of all the fragments, and that shorter superstrings should in general be more similar to the original.

\paragraph{Previous Results}

\begin{table}
\label{fig:previous-results}
\centering
\begin{tabular}{|l|c|c|}
\hline 
Authors & Date & Factor \\
\hline
Li~\cite{li} & 1990 & $O(\log(n))$\\ 
Blum, Jiang, Li, Tromp, Yannakakis~\cite{blum} & 1991 &  $3$\\
Teng, Yao~\cite{teng} & 1993 &  $2\frac{8}{9}$\\
Czumaj, Gasieniec, Piotr{\'o}w, Rytter~\cite{czumaj} & 1994 &  $2\frac{5}{6}$\\
Kosaraju, Park, Stein~\cite{kosaraju} & 1994 & $2\frac{50}{63}$\\ 
Armen, Stein~\cite{armen95} & 1995 & $2\frac{3}{4}$\\ 
Armen, Stein~\cite{armen96} & 1996 & $2\frac{2}{3}$\\
Breslauer, Jiang, Jiang~\cite{breslauer} & 1997 &  $2\frac{25}{42}$\\ 
Sweedyk~\cite{sweedyk} & 1999 &  $2\frac{1}{2}$ \\
Kaplan, Lewenstein, Shafrir, Sviridenko~\cite{kaplan} & 2005 & $2\frac{1}{2}$ \\
Paluch, Elbassioni, van Zuylen~\cite{paluch} & 2012 & $2\frac{1}{2}$ \\
\hline 
\end{tabular}
\caption{Previous results for the \ssp\ problem.}
\end{table}

Since \ssp\ is NP-hard~\cite{gallant,garey} and even MAX-SNP-hard~\cite{blum,vassilevska}, the best we can hope for in terms of approximation is a constant factor. A lot of effort went into designing approximation algorithms for the problem, Table~\ref{fig:previous-results} summarizes these developments. Note that the last two results, by Kaplan et al.~\cite{kaplan} and Paluch et al.~\cite{paluch} do not improve the approximation factor. They both give $\frac{2}{3}$-approximation algorithms for the related \tsp\ problem. Using a black-box reduction due to Breslauer et al.~\cite{breslauer}, these  give $2\frac{1}{2}$-approximation algorithms for \ssp. Both, especially the one due to Paluch et al., are significantly simpler than the original result of Sweedyk.  

Parallel to these developments, some progress has been made towards resolving the \emph{Greedy Superstring Conjecture} (see~\cite{storer,tarhio,turner}), which says that the greedy approach of repeatedly picking the two strings that overlap the most and gluing them together until only a single string remains, is actually a $2$-approximation. Blum et al.~\cite{blum} showed that the greedy algorithm gives a $4$-approximation, and Kaplan et al.~\cite{kaplan-greedy} improved this to $3\frac{1}{2}$.

\paragraph{Our Results/Techniques}

In this paper we develop several results that describe the structure of the overlaps of a collection of strings. Our results can be viewed as an extension of the framework introduced by Breslauer et al.~\cite{breslauer}. However, while Breslauer et al.\ use generic unbordered rotations and critical factorizations, we construct ours by using Lyndon words. It turns out that the added control we gain in this way allows for much more precise structural analysis of string overlaps.

We use these results to obtain a $2\frac{11}{23}$-approximation for \ssp, and therefore break a long-standing bound of $2\frac{1}{2}$.  

The basic idea of our approach is the following. For two strings $u$, $v$, let the \emph{overlap} of $u$ and $v$, denoted $\ov(u,v)$, be the longest suffix of $u$ which is also a prefix of $v$. The \emph{overlap graph} of a set of strings $S$ is a complete directed graph on $S$ with edge weights equal to lengths of corresponding overlaps. 

Blum et al.~\cite{blum} show how approximating \ssp\ for a set of strings $S$ can be reduced to approximating the problem of finding a longest path in the overlap graph of a certain auxiliary set of strings $R(S)$, called \emph{representative strings}. The performance of the resulting algorithm depends on how well we can bound the \emph{overlap loss} in the longest path approximation.

This bound can essentially be improved in two ways: by using a better approximation algorithm for the longest path problem in directed graphs (\tsp), or by providing a better bound on the overlap of the optimum path. For the first direction, Kaplan et al.~\cite{kaplan} and Paluch et al.~\cite{paluch} both give $\frac{2}{3}$-approximation for \tsp, which is the best known. For the second, the bounds given by Breslauer et al.~\cite{breslauer} are essentially tight.

In this paper we propose a third way to improve by joining the two objectives. Note that one can approximate \tsp\ by finding a maximum weight cycle cover, removing the lightest edge on each cycle, and then joining the resulting paths with arbitrary edges. This na\"{\i}ve algorithm only gives $\frac{1}{2}$-approximation, significantly weaker than $\frac{2}{3}$, the tight case being \emph{balanced} $2$-cycles. We observe however, that with a careful choice of representative strings $R(S)$, if the bounds given by Breslauer et al.\ are nearly tight, the cycles in the maximum weight cycle cover are far from balanced. So far in fact, that choosing the better of the two solutions: one given by a $\frac{2}{3}$-approximation algorithm, and one given by our na\"{\i}ve algorithm, results in an approximation algorithm for \ssp\ with ratio strictly smaller than $2\frac{1}{2}$.  

It is worth noting that, similarly to the approach of Breslauer et al., our algorithm is a black-box reduction from \ssp\ to \tsp. Therefore, any improvements on the approximation factor for the latter will yield an improvement for the former.

\paragraph{Organization of the Paper}

The paper is organized as follows. In Section~\ref{sec:prelim} we recall some facts regarding the properties of strings and their overlaps, as well as the standard approach to shortest string approximation. In Section~\ref{sec:approximation} we describe the new algorithm and analyse its approximation factor. This analysis relies on Theorem~\ref{thm:main-local}, which is the main technical result of this paper. The remaining part of the paper is devoted to proving this theorem.

In Section~\ref{sec:general-bounds} we present some general bounds concerning overlaps of strings. We believe they might be of independent interest. In Section~\ref{sec:proof} we use these bounds to prove the main theorem. Since the proof is a rather long and detailed case analysis, to facilitate understanding of the basic ideas of the paper, in Subsection~\ref{subsec:weaker} we give a simple proof of a weaker version of the main theorem. This version still gives an approximation factor smaller than $2\frac{1}{2}$.

Finally in Section~\ref{sec:tight} we show that Theorem~\ref{thm:main-local} is essentially tight. We also briefly discuss reasons why using our bounds to improve the analysis of the greedy algorithm might be difficult.

%%%%%%%%%%%%%%%%%%%%%%%%%%%%%%%%%%%%%%%%%%%%%%%%%%%%%%%%%%%%%%%%%%%%%%%%%%%%%%%%%%%%%%%%%%

\section{Preliminaries}
\label{sec:prelim}

In this section we recall some definitions, results and ideas concerning basic properties of strings. For a more extensive exposition the reader should consult any of the standard textbooks on combinatorics on words, e.g.\ the excellent monograph by Lothaire~\cite{lothaire}. 

We also describe the standard framework for \ssp\ approximation. Our presentation mostly follows that of Breslauer et al.~\cite{breslauer}. Note however, that instead of generic critical factorizations we use nice rotations, introduced at the end of Subsection~\ref{subsec:stringology}. This requires almost no changes in the framework, except for the proof of Lemma~\ref{lem:breslauer1}, which we provide.

\subsection{Stringology}
\label{subsec:stringology}

\paragraph{Basic concepts}
For a string $v$, we will use $v[i]$ to denote the $i$-th letter of $v$, and $v[i,j]$ to denote the substring
of $v$ consisting of letters $i,\ldots,j$. We will use $vu$ to denote concatenation of $v$ and $u$, and $v^k$ to denote the concatenation of $k$ copies of $v$. We will also use $v^\infty$ to denote the semi-infinite string $vvv\ldots$. Any representation of $w=uv$ as a concatenation of two (not necessarily nonempty) 
strings is called a \emph{factorization} of $w$. The factorization is \emph{nontrivial} if both $u$ and $v$ are nonempty.

For a string $w$ of length $n$, any integer $1 \le p \le n$ is a \emph{period} of $w$ if
$w[i] = w[i+p]$ for all $1 \le i \le n-p$. Note that $w$ always has at least one period, that
is its length. The smallest period of $w$ is called \emph{the minimum period} of $w$ or simply \emph{the period} of $w$, and denoted $p(w)$.

A string $w$ is \emph{primitive} if there is no $v$ such that $w=v^k$ with $k \ge 2$.

A string $z$ is called a \emph{rotation} of $w$ if there exists a factorization $w=uv$ such that $z=vu$. In that case we also say
that $z$ is a rotation \emph{starting at position} $|u|+1$, or that $|u|+1$ is $z$'s starting position in $w$. It is easy to see that
if $z$ is a rotation of $w$, then $z$ is primitive iff $w$ is. It is also a standard fact that if $w$ is primitive, then any rotation corresponding to a 
nontrivial factorization of $w$ is different from $w$. More generally, for a primitive $w$ and two different factorizations 
$w=u_1v_1$ and $w=u_2v_2$, the rotations $v_1u_1$ and $v_2u_2$ are different. It follows that for primitive $w$, every rotation of $w$ has a unique starting position in $w$. 

We say that two strings are \emph{equivalent} if one is a rotation of the other. Otherwise they are \emph{non-equivalent}.

We will assume a fixed order on the alphabet. This order induces a standard lexicographical order on the set of strings. 
We use $u \prec v$ to denote that $u$ is lexicographically smaller than $v$, and $u \preceq v$ to denote that $u$ is 
smaller or equal to $v$. 

Let $w$ be a primitive string and consider the order induced by $\prec$ on all rotatations of $w$. Let $\wmin$ and $\wmax$ 
be the minimal and maximal rotations in this order. Also, denote by $\imin(w)$ and $\imax(w)$ the starting positions of $\wmin$ and $\wmax$ in $w$.
Moreover, let $\hwmin$ and $\hwmax$ be strings such that $\wmin = \hwmin\hwmax$ and $\wmax = \hwmax\hwmin$.

\begin{lemma}
\label{lem:min-max}
Let $w$, $|w| \ge 2$ be a primitive string. Then both $\hwmin$ and $\hwmax$ are nonempty. In other words, $\wmin \neq \wmax$.
\end{lemma}

\begin{proof}
Since $|w| \ge 2$ and $w$ is primitive, it contains at least $2$ different letters. It follows that $\wmin$ and $\wmax$ start with different letters, and so they are different strings. In particular both $\hwmin$ and $\hwmax$ are nonempty.
\end{proof}

The following property of $\wmin$ and $\wmax$ is implicitly used by Crochemore et al.~\cite{crochemore}.

\begin{lemma}
\label{lem:unique}
Let $w$, $|w| \ge 2$ be a primitive string. Then $\wmax$ is the only rotation of $w$ that starts with $\hwmax$, and $\wmin$ is the only rotation that starts with $\hwmin$.
\end{lemma}

\begin{proof}
We will prove the claim for $\hwmax$, the other part of the proof is analogous. Suppose that there is a rotation $z=\hwmax v$ of $w$, other than $\wmax = \hwmax \hwmin$, that starts with $\hwmax$.
We claim that $z \succ \wmax$ which is a contradiction. To see that, notice that they both start with $\hwmax$, and $v \succeq \hwmin$ since $\hwmin$ is a prefix of the 
minimal rotation of $w$. The claim follows, since $v \neq \hwmin$.
\end{proof}

\paragraph{Borders}

A nonempty string $b$ is called a \emph{border} of a string $w$ if $w=bu=vb$ for some nonempty $u,v$.
A string is \emph{unbordered} if it has no border. So, a string $w$ is unbordered if $w$ has no proper prefix
that is also a suffix of $w$. The following is a standard fact (see e.g.\ Proposition 5.1.2 in Lothaire~\cite{lothaire}).

\begin{lemma}
\label{lem:unbordered}
Every primitive string $w$ has a rotation that is unbordered. In particular, $\wmin$ and $\wmax$ are unbordered. 
\end{lemma}

\begin{proof} 
Suppose that $\wmax$ has a border, i.e.\ there exists a proper prefix $v$ of $\wmax$ which is also its suffix. Let $|w|=|\wmax|=n$ and $|v|=k<n$. Since $v$ is a suffix of $\wmax$, we know that $u = (v\wmax)[1,n]$ is a rotation of $w$.
We claim that $u \succ \wmax$, which is a contradiction. To see this, notice that $u[1,k] = v = \wmax[1,k]$ and also $u[k+1,n] = \wmax[1,n-k] \succeq \wmax[k+1,n]$, since
$\wmax$ is maximal. So $u \succeq \wmax$, but this cannot be an equality since rotations of a primitive string are all different.

The same proof applies to $\wmin$ or one can simply notice that it is a maximal rotation in the lexicographical order induced by the reversed order on the alphabet.
\end{proof}

\begin{remark}
A primitive string $w$ such that $w = \wmax$ is called a \emph{Lyndon word}\footnote{We use term ``word'' here and not ``string'' as ``Lyndon word'' seems to be a well established phrase. In general, the algorithmic community tends to use the term ``string'' and the combinatorial community uses the term ``word''. We decided to follow this rule and use the term ``string'' with the single exception of ``Lyndon word''.} (w.r.t.\ the particular order on the alphabet that is used to define the lexicographical order). Note that the two rotations appearing in Lemma~\ref{lem:unbordered} are Lyndon words, and in fact this Lemma is equivalent to saying that Lyndon words are unbordered.
\end{remark}

\begin{remark}
One of the key ingredients of the results of Breslauer et al.~\cite{breslauer} is the notion of a critical factorization and 
,,The Critical Factorization Theorem'' (see C\'esari et al.~\cite{cesari}). Although we do not use them directly, a reader acquainted with Breslauer et al.~\cite{breslauer} will realize that they are nevertheless present in our work. In particular $\pmin(w)\pmax(w)$ and $\pmax(w)\pmin(w)$ are critical factorizations (this fact was used by Crochemore et al.~\cite{crochemore} in their proof of the Critical Factorization Theorem).
\end{remark}

\paragraph{Nice rotations}
Let $w$, $|w| \ge 2$ be a primitive string. \emph{The nice rotation} of $w$ is defined to be $\wmax$ if $|\hwmax| \le |\hwmin|$, otherwise it is defined to be $\wmin$. Let $\alpha(w) = \min(|\hwmax|,|\hwmin|)$. We will call a primitive string \emph{nice} if it is its own nice rotation. Note that if $w$ is nice, then:
\begin{itemize}
\item $w=\wmax$ and $\alpha(w) = |\hwmax| \le |w|/2$, or
\item $w=\wmin$ and $\alpha(w) = |\hwmin| < |w|/2$.
\end{itemize}
In particular we always have $\alpha(w) \le |w|/2$.

For a nice string $w$, we call $x$ a $w$-string if $x$ is a prefix of $w^\infty$.

\subsection{Shortest Superstring Approximation}

\paragraph{Basic ideas} 
In the remainder of this paper we assume w.l.o.g.\ that $S$ contains at least two strings and that no string in $S$ is a substring of another string

For two strings $u,v$ define the \emph{overlap} of $u$ and $v$, denoted $\ov(u,v)$, as the longest suffix of $u$ that is also a prefix of $v$. Also, define the \emph{prefix} of $u$ w.r.t.\ $v$, denoted $\pr(u,v)$, as the string $x$ such that $u=x\ov(u,v)$, i.e.\ prefix is the part of $u$ that does not overlap $v$. 

The following two directed graphs are good models of how the strings in $S$ overlap with each other. The \emph{overlap graph} of $S$ is a complete directed graph with $S$ as the vertex set, and edge 
$(s_i,s_j)$ having length $|\ov(s_i,s_j)|$. The \emph{prefix graph} (also called the \emph{distance graph}) 
is defined similarly, only edge $(s_i,s_j)$ now has length $|\pr(s_i,s_j)|$.

Let $\langle s_{i_1},s_{i_2},\ldots,s_{i_n}\rangle $ be the string 
$\pr(s_{i_1},s_{i_2})\pr(s_{i_2},s_{i_3})\ldots\pr(s_{i_{n-1}},s_{i_n})s_{i_n}$.
Obviously, it is the shortest string containing $s_{i_1},s_{i_2},\ldots,s_{i_n}$ in that order. 
Notice that the optimal solution has the form $\langle s_{i_1},s_{i_2},\ldots,s_{i_n}\rangle $ for some
ordering $s_{i_1},s_{i_2},\ldots,s_{i_n}$ of the strings in $S$.

The length of $\langle s_{i_1},s_{i_2},\ldots,s_{i_n}\rangle $ is equal to
\[|\pr(s_{i_1},s_{i_2})|+|\pr(s_{i_2},s_{i_3})|+\ldots+
|\pr(s_{i_{n-1}},s_{i_n})|+|\pr(s_{i_n},s_{i_1})|+|\ov(s_{i_n},s_{i_1})|,\]
which is the length of the cycle $s_{i_1}\rightarrow s_{i_2}\rightarrow\ldots\rightarrow s_{i_n}$ in the prefix graph of $S$ increased
by $|\ov(s_{i_n},s_{i_1})|$. Thus, the length of the shortest TSP tour in the prefix graph of 
$S$ lowerbounds the length of the shortest superstring.

The above considerations suggest that reduction to asymmetric TSP might be useful in approximating \ssp.
Unfortunately, the best known approximation algorithm for asymmetric TSP has factor  $O\left(\frac{\log n}{\log\log n}\right)$ (see~\cite{madry}),
so this approach is not very useful.

Let us look again at a generic solution $\langle s_{i_1},s_{i_2},\ldots,s_{i_n}\rangle $ and this time express its
length in terms of the overlap graph:
\[ |\langle s_{i_1},s_{i_2},\ldots,s_{i_n}\rangle| = \sum_{j=1}^n |s_j| - \sum_{j=1}^{n-1} |\ov(s_{i_j},s_{i_{j+1}})|.\]
The right term in the above expression (the \emph{total overlap} of $\langle s_{i_1},s_{i_2},\ldots,s_{i_n}\rangle $) is 
the length of the path $s_{i_1},\ldots,s_{i_n}$ in 
the overlap graph, so the longest TSP path in the overlap graph corresponds to the optimal solution for \ssp. 
Longest TSP path in a directed graph (called \tsp) can be approximated within constant factor. Notice however, that this does not lead to a constant factor approximation for \ssp. The problem is that the total overlap of the optimal solution could be very large compared to its length. In that case even a very good approximation algorithm for total overlap might only give mediocre approximation for the length of the superstring.

\paragraph{Two-step reduction to \tsp}

We can avoid the problems described in the previous paragraph by using the following two-step approach introduced by Blum et al.~\cite{blum}:
\begin{enumerate}
\item Find a minimum cycle cover $\mathcal{C}_{\min}$ in the distance graph.
\item For each cycle $C \in \mathcal{C}_{\min}$ construct a representative string $R(C)$ containing all strings in $C$ as substrings, let $R=R(\mathcal{C}_{\min})=\{R(C): C \in \mathcal{C}_{\min}\}$.
\item Find a \ssp\ solution for $R$ by reducing to \tsp.
\end{enumerate}
The idea here is that the first step groups strings with large overlaps together, so that the overlaps of the strings in $R$ are relatively small, and then the last step actually gives good approximation. 

The following series of lemmas and definitions from Blum et al.~\cite{blum} and Breslauer et al.~\cite{breslauer} gives an idea of why this approach works. 

For any cycle $C=s_{i_1}\rightarrow s_{i_2}\rightarrow \ldots \rightarrow s_{i_k}$ in $\mathcal{C}_{\min}$ let $R(C) = \langle s_{i_1},s_{i_2},\ldots,s_{i_k},s_{i_1}\rangle $. Note two interesting features of this definition. First, depending on where we break the cycle we can start $R(C)$ with any of the strings $s_{i_1},\ldots,s_{i_k}$. Second, $R(C)$ is actually ``too long'' as it unnecessarily contains two copies of $s_{i_1}$ --- this will turn out useful later on. 

Let $OPT(S)$ and $OPT(R)$ be the lengths of optimal \ssp\ solutions for $S$ and $R$. 

\begin{lemma}[Follows from Lemma 2.6 of~\cite{breslauer}, also implicit in~\cite{blum}]
\label{lem:blum-bound}
\[OPT(R) \le 2 OPT(S).\]
\end{lemma}

For a cycle $C = s_{i_1}\rightarrow s_{i_2}\rightarrow \ldots \rightarrow s_{i_k}$ in the prefix graph define 
\[s(C) = \pr(s_{i_1},s_{i_2})\pr(s_{i_2},s_{i_3}),\ldots,\pr(s_{i_k},s_{i_1}).\] 
Then $|s(C)|$ is the length of $C$ and $s(C)$ essentially reads the prefixes along the cycle. The strings $s(C)$ for $C \in \mathcal{C}_{\min}$ have very interesting properties.

\begin{lemma}[Claims 3 and 5 in Blum et al.~\cite{blum} ]
The strings $s(C)$ all all primitive with $s(C) \ge 2$ and are all non-equivalent.
\end{lemma}

Let $w(C)$ be the nice rotation of $s(C)$ for every $C \in \mathcal{C}_{\min}$.
As we already mentioned, the representative strings $R(C)$ defined as earlier are unnecessary long. This can be used to prove the following.

\begin{lemma}[Special case of Lemma 5.1 in Breslauer et al.~\cite{breslauer}]
One can define the representative $R(C)$ for a cycle $s_{i_1} \rightarrow \ldots \rightarrow s_{i_k}$ so that:
\begin{itemize}
\item $R(C)$ is a substring of $\langle s_{i_j},s_{i_{j+1}},\ldots,s_{i_k},s_{i_1},\ldots,s_{i_j}\rangle $ for some $j$ (in particular $OPT(R) \le 2OPT(S)$ still holds),
\item $R(C)$ is a $w(C)$-word.
\end{itemize} 
\end{lemma}

Finally, we need to show that the strings $R(C)$ do not overlap too much. The lemma below is stated in a slightly more general fashion so that it can be used more easily later on.

\begin{lemma}[Implicit in the proof of Lemma 3.3 in Breslauer et al.~\cite{breslauer}]
\label{lem:breslauer1}
Let $w_1$ and $w_2$ be non-equivalent nice words and let $x_i$ be a $w_i$-word for $i=1,2$. Also let $\alpha_i = \alpha(w_i)$ and $l_i = |w_i|$. Then
$|\ov(x_1,x_2)| < l_1+\alpha_2$.
\end{lemma}

\begin{proof}
Assume for a contradiction that $|\ov(x_1,x_2)| \ge l_1 + \alpha_2$. Consider the string $z = \ov(x_1,x_2)[l_1+1,l_1+\alpha_2]$. We have $z = \ov(x_1,x_2)[1,\alpha_2] = w_2[1,\alpha_2]$, so we need to have $l_1 = kl_2$ for some $k$ because of Lemma~\ref{lem:unique}, which is impossible.

To see why, notice that if $l_1 = kl_2$ and $|\ov(x_1,x_2)| \ge l_1$, then either $w_1$ and $w_2$ are equivalent (if $k=1$) or $w_1$ is nonprimitive (if $k>1$).
\end{proof}

\begin{theorem}[Breslauer et al.~\cite{breslauer}]
Given $c$-approximation for \tsp, one can approximate \ssp\ with approximation factor of $3\frac{1}{2} - 1\frac{1}{2} c$.
\end{theorem}

\begin{proof}
Consider the string $s=\langle R(C_1),\ldots,R(C_k)\rangle$ that is the optimal solution for $R$. Let $R_{OV}$ be the total overlap of this string. Then by applying Lemma~\ref{lem:breslauer1} to every pair of consecutive strings we get
\[ R_{OV} \le \sum_{i=1}^{k-1} \left(|w(C_i)|+\alpha(w(C_{i+1}))\right) \le \frac{3}{2}\sum_{i=1}^k |w(C_i)| \le \frac{3}{2} OPT(S).\]
A $c$-approximation algorithm for \tsp\ can be used to obtain a solution with total overlap of $c R_{OV}$. The length of the resulting \ssp\ solution is therefore at most
\[ OPT(R) + \frac{3}{2}(1-c) OPT(S) \le 2OPT(S)+\frac{3}{2}(1-c)OPT(S) = \left(3\frac{1}{2} - 1\frac{1}{2} c\right)OPT(S).\] 
\end{proof}

Since $\frac{2}{3}$-approximation algorithms for \tsp\ are known we obtain the following
\begin{corollary}[Kaplan et al.~\cite{kaplan}, also Paluch et al.~\cite{paluch}]
There exists a $2\frac{1}{2}$-approximation algorithm for \ssp.
\end{corollary}

%%%%%%%%%%%%%%%%%%%%%%%%%%%%%%%%%%%%%%%%%%%%%%%%%%%%%%%%%%%%%%%%%%%%%%%%%%%%%%%%%%%%%%%%%%

\section{The Algorithm}
\label{sec:approximation}

In this section we give the new approximation algorithm and bound its approximation factor.

\paragraph{Description}
The algorithm we are going to analyse is very simple. It returns a solution $S_0$ which is the better of the following two solutions $S_1$, $S_2$:
\begin{itemize}
\item $S_1$ is obtained by using any algorithm that reduces \ssp\ to \tsp\ (e.g.\ one due to Kaplan et al.~\cite{kaplan} or Paluch et al.~\cite{paluch}),
\item $S_2$ is also obtained by reducing to \tsp, but this time we get the final solution by computing the maximum weight cycle cover in the overlap graph of $R$ and dropping the lightest edge from every cycle.
\end{itemize}

\paragraph{Analysis}

For any cycle $C$ in the overlap graph let $O_C$ be the total overlap of $C$, i.e.\ sum of the weights of its edges. Let $M_C$ be the minimum weight of an edge of $C$. Also, let $L_C$ be the sum of the periods of the strings in $C$, which is equal to the total length of the corresponding cycles in $\mathcal{C}_{\min}$. 

Let $|R|$ be the total length of the representative strings in $R$ and let $\mathcal{C}$ be the maximum weight cycle cover in the overlap graph of $R$. Note that $\sum_{C \in \mathcal{C}} L_C = w(\mathcal{C}_{\min})$. Moreover, let $O_\mathcal{C} = \sum_{C \in \mathcal{C}} O_C$, let $M_\mathcal{C} = \sum_{C \in \mathcal{C}} M_C$. Finally, let $c$ be the best known approximation ratio for \tsp.

\begin{lemma}
\label{lem:losses}
$|S_1| \le OPT(R) + (1-c)O_\mathcal{C}$ and $|S_2| \le OPT(R) + M_\mathcal{C}$.
\end{lemma}

\begin{proof}
We have $|R| - O_\mathcal{C} \le OPT(R)$ and $|S_2| \le |R|-O_\mathcal{C}+M_\mathcal{C}$ which proves the second part.

For the first, note that since $O_\mathcal{C} \ge |R|-OPT(R)$, any algorithm that approximates $|R|-OPT(R)$ with factor $c$ gives a solution of length at most
$|S_1| \le |R| - c(|R|-OPT(R)) = OPT(R)+(1-c)(|R|-OPT(R)) \le OPT(R) + (1-c)O_\mathcal{C}$.
\end{proof}

The main technical ingredient of this paper is the following theorem (we show in Section~\ref{sec:tight} that it is essentially tight).

\begin{theorem}[Main Theorem (local version)]
\label{thm:main-local}
For any cycle $C$ in the overlap graph, we have
\[ 2M_C + 7O_C \le 11L_C.\]
\end{theorem}

By summing over all cycles of $\mathcal{C}$ we obtain the following.

\begin{theorem}[Main Theorem (global version)]
\label{thm:main-global}
\[ 2M_\mathcal{C} + 7O_\mathcal{C} \le 11w(\mathcal{C}_{\min}) \le 11OPT(S).\]
\end{theorem}

The next two sections are devoted to the proof of Theorem~\ref{thm:main-local}. For now let us see what it implies.

\begin{corollary}
\label{cor:factor}
\ssp\ can be approximated with factor $\left(2 + \frac{11(1-c)}{9-2c} \right)$. In particular for $c=\frac{2}{3}$ we get $2\frac{11}{23}$-approximation.
\end{corollary}

\begin{proof}
It follows from Lemma~\ref{lem:losses} and Lemma~\ref{lem:blum-bound} that
\[ |S_0| \le 2OPT(S) + \min(M_\mathcal{C},(1-c)O_\mathcal{C}).\]
We can bound the second term as follows:
\[ \min(M_\mathcal{C},(1-c)O_\mathcal{C}) \le \frac{2-2c}{9-2c}M_\mathcal{C}+\frac{7}{9-2c} (1-c)O_\mathcal{C} \le \frac{11(1-c)}{9-2c}OPT(S).\]

For $c=\frac{2}{3}$ we obtain $\frac{11\cdot \frac{1}{3}}{9-2\cdot \frac{2}{3}} = \frac{11}{23}$ and so $|S_0| \le 2\frac{11}{23} OPT(S)$.
\end{proof}

%%%%%%%%%%%%%%%%%%%%%%%%%%%%%%%%%%%%%%%%%%%%%%%%%%%%%%%%%%%%%%%%%%%%%%%%%%%%%%%%%%%%%%%%%%

\section{The General Bounds}
\label{sec:general-bounds}

In this section we present and prove the bounds on overlaps of strings. We consider a set of non-equivalent nice strings $w_1,\ldots,w_k$, and for each $i=1,\ldots,k$ a $w_i$-string $x_i$. We use $l_i$ to denote $|w_i|$, and $\alpha_i$ to denote $\alpha(w_i)$. Moreover, for each $i\neq j$, $i,j \in \{1,\ldots,k\}$ we define $\ov_{ij} = \ov(x_i,x_j)$ and $o_{ij}=|\ov_{ij}|$. Finally, let $w_{ij}$ be the rotation of $w_i$ that matches $\ov(x_i,x_j)$ from the left. If there is more than one such rotation (which might happen if $o_{ij} < l_i$), choose any such rotation.

By Lemma~\ref{lem:breslauer1} we have $o_{12} \le l_1 + \frac{1}{2}l_2$. The main theme of this section is characterizing situations in which  this inequality is in some way non-tight. The underlying idea in most (but not all) of these results is the following:  We show that if $o_{12}$ is actually close to its upper-bound, then the set of possible starting positions of the maximal/minimal rotation of $w_1$ is strongly limited, which in turn leads to an upper-bound on $\alpha_1$. This can be used to upper-bound other overlaps using Lemma~\ref{lem:breslauer1}.

We start with another lemma from the work of Breslauer et al.~\cite{breslauer}.
\begin{lemma}[Implicit in the proof of Lemma 3.3 of Breslauer et al.~\cite{breslauer}]
\label{lem:breslauer2}
If $l_1 \le l_2$ then $o_{12} < l_2$. For general $l_1,l_2$ we have $o_{12} < kl_2$ whenever $l_1 \le kl_2$.
\end{lemma}

\begin{proof}
Assume for a contradiction that $l_1 \le kl_2$ and $o_{12} \ge kl_2$. Also, w.l.o.g.\ assume that $k$ is the smallest integer such that $l_1 \le kl_2$. 

Similarly as in the proof of Lemma~\ref{lem:breslauer1} we cannot have $l_1 = kl_2$, and we also cannot have $l_1 = (k-1)l_2$ for the same reasons. Therefore $(k-1)l_2 < l_1 < kl_2$. 

Consider now the string $\ov_{12}[l_1+1,kl_2]=\ov_{12}[1,kl_2-l_1]$. This string is a non-trivial suffix of $w_2$, and also a prefix of $w_2$, a contradiction with $w_2$ being nice and Lemma~\ref{lem:unbordered}. 
\end{proof}

The next two lemmas demonstrate that $o_{12}$ getting close to $l_1+\frac{1}{2}l_2$ implies an upper-bound on the value of $\alpha_1$. While Lemma~\ref{lem:general-long-l1-cor} gives this bound explicitly, Lemma~\ref{lem:general-long-l1} describes it in terms of constraints on the starting positions of maximal and minimal rotations of $w_1$.

\begin{lemma}
\label{lem:general-long-l1}
Let $l_1 \ge l_2$, $o_{12} \ge l_2$ and let $w_2$ be its maximal rotation, then:
\begin{itemize}
\item $\imax(w_{12}) = 1$, or $\imax(w_{12}) = l_2 \lfloor \frac{o_{12}-1}{l_2}\rfloor+1$ or $\imax(w_{12}) > \max(l_2 \lfloor \frac{o_{12}-1}{l_2}\rfloor+1,o_{12}-\alpha_2+1)$.
\item $\imin(w_{12}) = \alpha_2+1$, or $\imin(w_{12}) =  l_2 \lfloor \frac{o_{12}-\alpha_2-1}{l_2}\rfloor+\alpha_2+1$, or  $\imin(w_{12}) > \max(l_2 \lfloor \frac{o_{12}-\alpha_2-1}{l_2}\rfloor+\alpha_2+1,o_{12}-(l_2-\alpha_2)+1)$,
\end{itemize}
Moreover, if $o_{12} \ge l_1$, then $\imax(w_{12}) \neq 1$ and $\imin(w_{12}) = \alpha_2+1$.
\end{lemma}

\begin{remark}
Several of the lemmas appearing in the remainder of this section assume that one of the strings involved is its maximal rotation. In all cases a symmetrical statement is true as well, in which the roles of minimal and maximal rotations of all strings involved are reversed. We omit the corresponding statements in all these lemmas.
\end{remark}

\begin{proof}[Proof (of Lemma~\ref{lem:general-long-l1})]
Let us start with claims concerning $\imax(w_{12})$. Since we assume $o_{12} \ge l_2$, we know that $w_{12}$ contains $\pmax(w_2)$. Therefore $(w_{12})_{\max}[1,\alpha_2] \succeq \pmax(w_2)$. We now consider two cases, depending on whether or not $\imax(w_{12}) \le o_{12}-\alpha_2+1$.

\case{1} If $\imax(w_{12}) \le o_{12}-\alpha_2+1$, then $\imax(w_{12})+\alpha_2-1 \le o_{12}$ and consequently $(w_{12})_{\max}[1,\alpha_2] \preceq \pmax(w_2)$. By the previous observation this is in fact an equality, and by Lemma~\ref{lem:unique} $\imax(w_2) = kl_2 + 1$ for some natural $k$, i.e.\ the maximal rotation of $w_{12}$ is aligned with the starting position of some occurence of $w_2$ in $o_{12}$. The two positions that appear in the statement of the lemma: $\alpha_2+1$ and $l_2 \lfloor \frac{o_{12}-1}{l_2}\rfloor+1$, are the starting positions of the leftmost and the rightmost occurences, respectively. We will show that $\imax(w_{12})$ has to be equal to one of them.

Suppose that this is not the case. This means that we have $\imax(w_{12}) = kl_2 + 1$ and both $(k-1)l_2+1$ and $(k+1)l_2+1$ are in $[1,\ldots,o_{12}]$. Note that by Lemma~\ref{lem:breslauer2} we then also have $(k+1)l_2 \le l_1$, and in fact $(k+1)l_2 < l_1$ since otherwise $l_1$ would not be primitive. Therefore $(k+1)l_2+1 \le l_1$.
Consider the rotations $r_1,r_2,r_3$ of $w_{12}$ starting at positions $(k-1)l_2+1$, $kl_2+1$ and $(k+1)l_2+1$, respectively. We have $r_1 = w_2w_2w$, $r_2 = w_2ww_2$ and $r_3=ww_2w_2$ for some string $w$.
Since all rotations of a primitive string are different, we have $ww_2 \neq w_2w$. If $ww_2 \succ w_2w$, then $r_3$ is the largest of the three rotations. If, on the other hand, $ww_2 \prec w_2w$, then $r_1$ is the largest one. Therefore $\imax(w_{12}) \neq kl_2+1$, a contradiction.

\case{2} We are left with the case where $\imax(w_{12}) > o_{12}-\alpha_2+1$, and we can also assume that we do not have $\imax(w_{12})=l_2 \lfloor \frac{o_{12}-1}{l_2}\rfloor+1$ since then the lemma clearly holds.

We need to prove that $\imax(w_{12}) > l_2 \lfloor \frac{o_{12}-1}{l_2}\rfloor+1$. If that was not the case, then $(k-1)l_2 + 1 < \imax(w_{12}) \le kl_2 < o_{12}$ for $k=\lfloor \frac{o_{12}-1}{l_2}\rfloor$. Then $w_{12}[\imax(w_{12}),kl_2]$ is a non-trivial suffix of $w_2$, and it is also a prefix of $w_2$ by maximality of the rotation of $w_{12}$ starting at $\imax(w_{12})$. But that is a contradiction with the fact that $w_2$ is unbordered by Lemma~\ref{lem:unbordered}. This ends the proof of the bounds for $\imax(w_{12})$.

One final claim we need to show concerning $\imax(w_{12})$ is that if $o_{12} \ge l_1$, then $\imax(w_{12}) \neq 1$. Since for $o_{12} \ge l_1$ we have $l_1 > l_2$, there are at least two positions of the form $kl_2+1$ within $w_{12}$. Consider rotations $r_1$ and $r_2$ of $w_{12}$ starting at two consecutive such positions $kl_2+1$ and $(k+1)l_2+1$. We will prove that $r_1 \prec r_2$, which implies our claim. We have $r_1=w_2 w w_2^k$ and $r_2 =ww_2^{k+1}$ for some $w$ which is a prefix of $w_2^\infty$. In particular, we have $r_1 = wvw_2^k$ for some $v$ which is a rotation of $w_2$. Since $r_1 \neq r_2$ by Lemma~\ref{lem:unique}, and $w_2$ is its maximal rotation, we conclude that $r_1 \prec r_2$. 

Let us now prove the claims concerning $\imin(w_{12})$.
Similarly to the case of $\imax(w_{12})$ we can argue that either $\imin(w_{12}) > o_{12} - (l_2-\alpha_2)+1$ or we have $\imin(w_{12}) = kl_2 + \alpha_2 + 1$ for some $k$. 

This time it will be more convenient to start with the case of $o_{12} \ge l_1$. Among rotations starting at positions of the form $kl_2+\alpha_2+1$, the one starting at $\alpha_2+1$ is minimal in this case, and the proof is almost identical to the one we just presented for $\imax(w_{12})$. 

So we only need to exclude the case where $o_{12} \ge l_1$ and $\imin(w_{12}) > o_{12}-(l_2-\alpha_2)+1$. If that happened, then we would have $(w_{12})_{\min}[1,l_2-\alpha_2] = w_{12}[\imin(w_{12}),l_1] w$, where $w$ is a non-empty prefix of $w_2$. Call this string $r_1$ and let $r_2=w_{12}[\alpha_2+1,l_2]$. We claim that $r_2  \prec r_1$, which is a contradiction with minimality of $r_1$. To see that $r_2 \prec r_1$, note that $w_{12}[\alpha_2+1,\alpha_2+l_1-\imin(w_{12})] \preceq w_{12}[\imin(w_{12}),l_1]$ by the definition of $\alpha_2$. Moreover $w_{12}[l_2-|w|+1,l_2] \preceq w$ because $w$ is a prefix of $(w_2)_{\max}$. But we cannot have an equality here, since then $w$ would also be a suffix of $(w_2)_{\max}$, a contradiction with Lemma~\ref{lem:unbordered}.

Let us now prove the main claims concerning $\imin(w_{12})$. Again, we consider two cases.

\case{1} If $\imin \le o_{12}-(l_2-\alpha_2)+1$ and consequently $\imin(w_{12})$ is of the form $kl_2+\alpha_2+1$, then we can show that we have either $\imin(w_{12})=\alpha_2+1$ or $\imin(w_{12}) = 
l_2 \lfloor \frac{o_{12}-\alpha_2-1}{l_2}\rfloor+\alpha_2+1$. The proof  is almost identical to the one we provided for $\imax(w_{12})$. The only step that does not directly translate, is that the rightmost position of the form $kl_2+\alpha_2+1$ within $\ov_{12}$ is also within $w_{12}$. Luckily, we already considered the case of $o_{12} \ge l_1$.

\case{2} If $\imin(w_{12}) > o_{12} - (l_2-\alpha_2)+1$, then we can also assume that we do not have $\imin(w_{12})= l_2 \lfloor \frac{o_{12}-\alpha_2-1}{l_2}\rfloor+\alpha_2+1$. We need to show that $\imin(w_{12}) > l_2 \lfloor \frac{o_{12}-\alpha_2-1}{l_2}\rfloor+\alpha_2+1$. Again, the proof is almost identical to the one for $\imax(w_{12})$.
\end{proof}

\begin{lemma}
\label{lem:general-long-l1-cor}
Let $l_1,o_{12} \ge l_2$ and let $w_2$ be its maximal rotation. Then we always have $\alpha_1 \le l_2 + (l_1 + \alpha_2 - o_{12})$ and moreover: 
\begin{enumerate}
\item either $\alpha_1 \le l_1 + l_2 - o_{12}$, or
\item the maximal rotation of $w_{12}$ starts at the rightmost position of the form $kl_2+1$, or the minimal rotation of $w_{12}$ starts at the rightmost position of the form $kl_2+\alpha_2+1$.
\end{enumerate}
\end{lemma}

\begin{proof}
Since we assume $o_{12} \ge l_2$, Lemma~\ref{lem:general-long-l1} describes all posibilities for $\imin(w_{12})$ and $\imax(w_{12})$. The rest is simple case analysis. 

If either $\imax(w_{12}) = l_2 \lfloor \frac{o_{12}-1}{l_2}\rfloor + 1$ or $\imin(w_{12}) = l_2 \lfloor \frac{o_{12}-\alpha_2-1}{l_2}\rfloor+\alpha_2+1$ (i.e.\ the second alternative in the statement of the lemma holds) then $\imin(w_{12})$ is either at most $\alpha_2+1$ or at least $o_{12}-l_2+1$ by Lemma~\ref{lem:general-long-l1}, and the same bounds hold for $\imax(w_{12})$. ``Wrapping around'' the end of $w_{12}$, they both land in an interval of length $l_1-(o_{12}-l_2)+\alpha_2 = l_2 + (l_1+\alpha_2-o_{12})$ and hence this quantity is also an upper bound on $\alpha_1$.

If neither $\imax(w_{12}) = l_2 \lfloor \frac{o_{12}-1}{l_2}\rfloor + 1$ nor $\imin(w_{12}) = l_2 \lfloor \frac{o_{12}-\alpha_2-1}{l_2}\rfloor+\alpha_2+1$, then Lemma~\ref{lem:general-long-l1} gives even stronger bounds. We have $\imin(w_{12}) \le \alpha_2 + 1$ or $\imin(w_{12}) > o_{12}-(l_2-\alpha_2)+1$ and the same bounds (in fact stronger) hold for $\imax(w_{12})$. Repeating the previous argument we get 
\[ \alpha_1 \le l_1-(o_{12}-(l_2-\alpha_2)) + \alpha_2 = l_1 + l_2 - o_{12}.\]
\end{proof}

The next two lemmas state some of the consequences of Lemmas~\ref{lem:general-long-l1} and~\ref{lem:general-long-l1-cor}, that are particularly easy to use.

\begin{lemma} 
\label{lem:delta}
If $l_1 \ge l_2$ then:
\begin{enumerate}
\item $o_{12} + \alpha_1 \le l_1 + l_2$ for $l_1 < 2l_2$,
\item $o_{12} + \alpha_1 \le 2l_1-l_2 = l_1+l_2 + (l_1-2l_2)$ for $2l_2 \le l_1 < \frac{5}{2}l_2$,
\item $o_{12} + \alpha_1 \le l_1 + l_2+\alpha_2$.
\end{enumerate}
\end{lemma}

\begin{proof}
In the proof, we assume w.l.o.g.\ that $w_2$ is its maximal rotation.

Let us first consider the case of $l_1 < 2l_2$. If $o_{12} < l_1$ then we get the claim, since $\alpha_1 \le \frac{1}{2}l_1 \le l_2$. On the other hand, if $o_{12} \ge l_1$ (note that in this case $l_1 > l_2$), then  by Lemma~\ref{lem:general-long-l1} we have $\imin(w_{12}) = \alpha_2 + 1$ and $\imax(w_{12}) \in \{l_2+1\} \cup (o_{12}-\alpha_2+1,\ldots,l_1]$. Therefore, we either have $\alpha_1 \le l_2-\alpha_2$ or $\alpha_1 \le (l_1-o_{12}+\alpha_2)+\alpha_2$ and in both cases it is easy to verify that our claim is true.

To prove the second inequality, we consider three cases: 

\case{1} If $o_{12} < l_1$ then
\[ o_{12} + \alpha_1 \le l_1 + \alpha_1 \le \frac{3}{2}l_1 \le 2l_1 - l_2.\]

\case{2} If the first alternative in Lemma~\ref{lem:general-long-l1-cor} holds, i.e.\ $\alpha_1 \le l_1 + l_2 - o_{12}$ then
\[ o_{12}+ \alpha_1 \le l_1 + l_2 \le 2l_1 - l_2.\]

\case{3} We are left with the case where the second alternative of Lemma~\ref{lem:general-long-l1-cor} holds. Since $o_{12} \ge l_1$ and so $\imin(w_{12}) = \alpha_2 + 1$, this means that $\imax(w_{12}) = 2l_2+1$. It follows that $\alpha_1 \le (l_1-2l_2) + \alpha_2$ and so
\[ o_{12} + \alpha_1 \le (l_1 + \alpha_2) + (l_1-2l_2) + \alpha_2 \le 2l_1-l_2.\] 

The third inequality of Lemma~\ref{lem:delta} follows immediately from the inequality $\alpha_1 \le l_2 + (l_1+\alpha_2-o_{12})$ in the first part of Lemma~\ref{lem:general-long-l1-cor}.
\end{proof}

Let $\Delta o_{ij} = (l_i+\frac{1}{2}l_j)-o_{ij}$ and $\Delta \alpha_i = \frac{1}{2} l_i - \alpha_i$. These basically measure how much smaller $o_{ij}$ and $\alpha_i$ are from their maximum values.

\begin{corollary}
\label{cor:delta}\ 
\begin{enumerate}
\item $\Delta o_{12} + \Delta \alpha_1 \ge \frac{1}{2}(l_1 - l_2)$ if $l_1 < 2l_2$, 
\item $\Delta o_{12} + \Delta \alpha_1 \ge \frac{1}{4}(l_1 - l_2)$ if $l_1 \ge 3l_2$, and
\item $\Delta o_{12} + \Delta \alpha_1 \ge \frac{1}{6}(l_1 - l_2)$.
\end{enumerate}
\end{corollary}

\begin{proof}
For the first inequality, we have by Lemma~\ref{lem:delta}
\[ \Delta o_{12} + \Delta \alpha_1 = \frac{3}{2}l_1 + \frac{1}{2}l_2 - o_{12} - \alpha_1 \ge \frac{3}{2}l_1 + \frac{1}{2}l_2 - l_1 - l_2  = \frac{1}{2}(l_1-l_2).\]

For the second inequality, we have by Lemma~\ref{lem:delta}
\[ \Delta o_{12} + \Delta \alpha_1 \ge \frac{3}{2}l_1 + \frac{1}{2}l_2 - (l_1+l_2+\alpha_2) = \frac{1}{2}l_1 - \frac{1}{2}l_2 - \alpha_2 \ge \frac{1}{2}l_1 - l_2 \ge \frac{1}{4}l_1 + \frac{3}{4}l_2 - l_2 = \frac{1}{4}(l_1-l_2).\]

Clearly, we only need to prove the third inequality for $2l_2 \le l_1 < 3l_2$. We consider two cases:

\case{1} If $2l_2 \le l_1 \le \frac{5}{2}l_2$ then
\[ \Delta o_{12} + \Delta \alpha_1 \ge \frac{3}{2}l_1 + \frac{1}{2}l_2 - (2l_1-l_2)  = \frac{3}{2}l_2 - \frac{1}{2}l_1 = \frac{1}{6}(l_1-l_2) + \frac{10}{6}l_2 - \frac{4}{6}l_1 \ge \frac{1}{6}(l_1-l_2).\]

\case{2} If $\frac{5}{2}l_2 \le l_1 < 3l_2$ then we have
\[ \Delta o_{12} + \Delta \alpha_1 \ge \frac{3}{2}l_1 + \frac{1}{2}l_2 - (l_1+l_2+\alpha_2)   \ge \frac{1}{2}l_1 - l_2 \ge \frac{1}{6}l_1 + \frac{5}{2} \cdot \frac{1}{3} l_2 - l_2 = \frac{1}{6}(l_1-l_2).\]
 
\end{proof}

The last lemma is not used in the proof of the main theorem. Nevertheless, we decided to include it in this section, as we believe it might turn out useful in future developments.

\begin{lemma}
\label{lem:general-short-l1}
Let $l_1 \le l_2$ and let $w_2$ be its maximal rotation. If $o_{12} \ge l_1 + \alpha_2 - \alpha_1$, then $\alpha_1 \le |\alpha_2 - kl_1|$ for all positive integers $k$.
\end{lemma}

\begin{remark}
The most important consequence of the above lemma is that if $l_2 \sim 2l_1$, then we cannot have $\alpha_1 \sim \frac{1}{2}l_1$, $\alpha_2 \sim \frac{1}{2}l_2$ and $o_{12} \sim l_1 + \alpha_2$ all happening at the same time.
\end{remark}

\begin{proof}
We have $o_{12} \ge l_1 - \alpha_1 \ge |\pmax(w_{12})|$, so $w_2[1,|\pmax(w_{12})|] = (w_2)_{\max}[1,|\pmax(w_{12})|]$ is a substring of $\ov_{12}$, and in particular $\pmax(w_{12}) \succeq (w_2)_{\max}[1,|\pmax(w_{12})|]$. Therefore, if $\pmax(w_{12})$ is contained in $\ov_{12}$ we need to have $\pmax(w_{12}) = w_{12}[1,|\pmax(w_{12})|]$ and $\imax(w_{12})=1$ by Lemma~\ref{lem:unique}. If, on the other hand, $\pmax(w_{12})$ is not contained in $\ov_{12}$, then we have $\imax(w_{12}) > o_{12}-(l_1-\alpha_1)+1 \ge \alpha_2+1$ and so $\imax(w_{12}) \in [\alpha_2+1,\ldots,l_1]$.

Similarly we can see that $\ov_{12}$ contains $w_2[\alpha_2+1,\alpha_2+|\pmin(w_{12})|]$, and so if $\pmin(w_{12})$ is contained in $\ov_{12}$, we have $\pmin(w_{12}) = x_2[\alpha_2+1,\alpha_2+|\pmin(w_{12})|]$, and by Lemma~\ref{lem:unique} $\imin(w_{12}) = (\alpha_2+1) \mod l_1$. Otherwise, we have $\imin(w_{12}) \in [\alpha_2+1,\ldots,l_1]$.

It is easy to verify that in all cases for $\imax(w_{12})$ and $\imin(w_{12})$ we get $\alpha_1 \le |\alpha_2-kl_1|$ for all positive integers $k$. 
\end{proof}

\section{The Proof of the Main Theorem}
\label{sec:proof}

In this section we present the proof of Theorem~\ref{thm:main-local}. We first introduce some additional definitions and technical lemmas, designed specifically for this proof, in Subsection~\ref{subsec:proof-prelims}. Since the proof itself is a rather long and detailed case analysis, in Subsection~\ref{subsec:weaker} we present a simple proof of a weaker version of Theorem~\ref{thm:main-local}. This weaker statement still gives an approximation ratio below $2\frac{1}{2}$. The proof of Theorem~\ref{thm:main-local} follows, for easier reading split into a subsection covering some basic observations and four subsections corresponding to different cycle lengths.

\subsection{Preliminaries}
\label{subsec:proof-prelims}

We keep the notation from previous chapters. In particular, for a cycle $C=x_1 \rightarrow x_2 \rightarrow \ldots \rightarrow x_k \rightarrow x_1$, we are interested in bounding 
$M=M_C = \min\{o_{ij} : (x_i,x_j) \in C\}$ and $O=O_C = \sum_{(x_i,x_j) \in C} o_{ij}$ in terms of $L=L_C=\sum_{i=1}^k l_i$. Recall also, that $\Delta o_{ij} = (l_i+\frac{1}{2}l_j)-o_{ij}$ and $\Delta \alpha_i = \frac{1}{2} l_i - \alpha_i$. Let $\Delta O = \sum_{(x_i,x_j) \in C} \Delta o_{ij} = \frac{3}{2}L - O$. 

We now introduce a couple more definitions. We call an edge $x_i \rightarrow x_j$ a \emph{down-edge} if $l_i \ge l_j$ and we call it an \emph{up-edge} otherwise. We denote the sets of down-edges and up-edges of $C$ by $C_d$ and $C_u$ respectively. A down-edge $x_i \rightarrow x_j$ is \emph{steep} if $l_i \ge 2l_j$, otherwise it is \emph{flat}. Similarly an up-edge $x_i \rightarrow x_j$ is \emph{steep} if $l_i \le \frac{1}{2} l_j$, and \emph{flat} otherwise.

Finally let $\lmin$ and $\lmax = l_1$ be the smallest and the largest among $l_1,\ldots,l_k$ breaking ties arbitrarily.

\begin{lemma}
\label{lem:delta-up}
For any up-edge $x_i\rightarrow x_j$ we have
\[ \Delta o_{ij} \ge l_i - \frac{1}{2}l_j \ge \lmin - \frac{1}{2}\lmax.\]
\end{lemma}

\begin{proof}
The second inequality is obvious. 

As for the first, there is nothing to prove for steep up-edges since then $l_i - \frac{1}{2}l_j \le 0$. For flat up-edges we have
\[ \Delta o_{ij} \ge \Big(l_i + \frac{1}{2}l_j\Big) - l_j = l_i - \frac{1}{2}l_j,\]
by Lemma~\ref{lem:breslauer2}. 
\end{proof}

\begin{lemma}
\label{lem:delta-app}
For any cycle $C$ we have
\[\Delta O \ge \frac{1}{12}\sum_{(x_i,x_j) \in C_d} \Big(l_i-l_j\Big) \ge \frac{1}{12}\Big(\lmax-\lmin\Big).\]
The constant can be improved to $\frac{1}{8}$ if there are no two consecutive steep down-edges in $C$, and to $\frac{1}{4}$, if there are no steep down-edges in $C$.
\end{lemma}

\begin{proof}
Let $x_i \rightarrow x_j$ be a down-edge, and let $x_l \rightarrow x_i$ be the edge preceding it on $C$. Then we get from Corollary~\ref{cor:delta}:
\begin{equation}
\label{eqn:sum-delta}
\Delta o_{li} + \Delta o_{ij} \ge \Delta \alpha_i + \Delta o_{ij} \ge \frac{1}{6}\Big(l_i-l_j\Big)
\end{equation} 
The right-hand side of the sum of inequality~\eqref{eqn:sum-delta} over all down-edges is upper-bounded by $2\Delta O$ and the claim follows. 

If there are no steep down-edges in $C$, then this reasoning can be repeated using the sharper bound in Corollary~\ref{cor:delta}.

Finally, if there are no two consecutive steep down-edges in $C$, then let $C_s$ be the set of steep down-edges and consider the sum of inequality~\eqref{eqn:sum-delta} over $C_s$. Since steep down-edges are nonconsecutive, the right-hand side of this inequality is upperbounded by $\Delta O$, and so
\[ \Delta O \ge \frac{1}{6} \sum_{(x_i,x_j) \in C_s} \Big(l_i - l_j\Big) .\]
We can also slightly improve the first part of the proof by using a stronger bound for flat edges to obtain:
\[ 2\Delta O \ge \frac{1}{6} \sum_{(x_i,x_j) \in C_s} \Big(l_i - l_j\Big) + \frac{1}{2} \sum_{(x_i,x_j) \in C_d\setminus C_s} \Big(l_i - l_j\Big).\]
Adding twice the first inequality to the second one, we get
\[ 4\Delta O \ge \frac{1}{2} \sum_{(x_i,x_j) \in C_s} \Big(l_i - l_j\Big) + \frac{1}{2} \sum_{(x_i,x_j) \in C_d\setminus C_s} \Big(l_i - l_j\Big),\]
and the claim follows.
\end{proof}

\begin{lemma}
\label{lem:delta-flat}
If $\lmin > \frac{1}{2} \lmax$ then
\[ \Delta O \ge \frac{1}{4} \lmin.\]
\end{lemma}

\begin{proof}
Note that in the proof of Lemma~\ref{lem:delta-app} we actually have
\[ 2\Delta O \ge \frac{1}{6} \sum_{(x_i,x_j) \in C_s} \Big(l_i - l_j\Big) + \frac{1}{2} \sum_{(x_i,x_j) \in C_d\setminus C_s} \Big(l_i - l_j\Big) + \sum_{(x_i,x_j) \in C_u} \Delta o_{ij}.\]
If $\lmin > \frac{1}{2} \lmax$, then since all edges are flat and there is at least one down-edge (we excluded the case of all $l_i$ equal) we obtain
\[ 2\Delta O \ge \frac{1}{2} \sum_{(x_i,x_j) \in C_d} \Big(l_i - l_j\Big) + \Big(\lmin-\frac{1}{2}\lmax\Big) \ge \frac{1}{2}\Big(\lmax-\lmin\Big)+\Big(\lmin-\frac{1}{2}\lmax\Big) = \frac{1}{2}\lmin,\]
by Lemma~\ref{lem:delta-up} and the claim follows.
\end{proof}

\subsection{Proof of a Weaker Version of the Main Theorem}
\label{subsec:weaker}

In this subsection we present a weaker version of the Theorem~\ref{thm:main-local}, which is relatively easy to prove, and still leads to approximation factor smaller than $2\frac{1}{2}$.

\begin{theorem}[Main Theorem, weak local version]
For any cycle $C$ in the overlap graph of $R$ we have
\[ M_C + 24O_C \le 36\frac{1}{4}L_C.\]
\end{theorem}

\begin{remark}
Note that the emphasis here is on simplicity, and the proof below can easily be improved in many ways.
\end{remark}

\begin{proof}
We first prove that we always have $\Delta O \ge \frac{1}{24k}L$, where $k$ is the length of the cycle. We consider two cases:

\case{1} If $\lmax < 2\lmin$ then by Lemma~\ref{lem:delta-flat} we have $\Delta O \ge \frac{1}{4}\lmin$ and so 
\[ \Delta O \ge \frac{1}{4}\lmin \ge \frac{1}{4(2k-1)} \Big(2k-1\Big) \lmin \ge \frac{1}{4(2k-1)}\Big(\lmin+(k-1)\lmax\Big) \ge \frac{1}{4(2k-1)}L \ge \frac{1}{24k}L.\]

\case{2} If $\lmax \ge 2\lmin$ then we have by Lemma~\ref{lem:delta-app} 
\begin{align*} \Delta O \ge \frac{1}{12}\Big(\lmax-\lmin\Big) \ge \frac{1}{12(2k-1)}\Big((2k-1)\lmax - (2k-1)\lmin\Big) \ge  \\
\ge \frac{1}{12(2k-1)}\Big((k-1)\lmax + \lmin\Big) \ge \frac{1}{12(2k-1)}L \ge \frac{1}{24k}L.
\end{align*}

From the inequality we have just proved we get
\[ O \le \left(\frac{3}{2} - \frac{1}{24k}\right)L.\]
We also always have 
\[ M \le \frac{1}{k} \cdot O \le \frac{3}{2k} L.\] 
Joining these gives
\[ M + 24O \le \left( \frac{3}{2k} + 36 - \frac{1}{k}\right)L = \left(36 + \frac{1}{2k}\right)L \le 36\frac{1}{4} L.\] 
\end{proof}

\begin{corollary}
There exists a $2\frac{145}{292}$-approximation algorithm for \ssp. 
\end{corollary}

\begin{proof}
Similarly as in the proof of Corollary~\ref{cor:factor}, we get from Lemma~\ref{lem:losses} and Lemma~\ref{lem:blum-bound} that
\[ |S_0| \le 2OPT(S) + \min\Big(M,\Big(1-\frac{2}{3}\Big)O\Big).\]
We can bound the second term as follows:
\[ \min\Big(M,\frac{1}{3}O\Big) \le \Big(\frac{1}{73}M + \frac{72}{73}\cdot\frac{1}{3}O\Big) \le \frac{36\frac{1}{4}L}{73} = \frac{145}{292}L \le \frac{145}{292}OPT(S),\]
and so
$|S_0| \le 2\frac{145}{292}OPT(S)$.
\end{proof}

\subsection{The Proof of Theorem~\ref{thm:main-local} -- Basic Observations} 

Let us recall Theorem~\ref{thm:main-local}.

{
\renewcommand{\thetheorem}{\ref{thm:main-local}}
\begin{theorem}[restated]
For every cycle $C$ in the overlap graph of $R$, we have \[2M_C+7O_C \le 11L_C.\] 
\end{theorem}
\addtocounter{theorem}{-1}
}

We can easily get rid of the following special case, which will make some reasonings easier later on.
\begin{lemma}
\label{lem:all-equal}
If all $l_i$ are equal for a cycle $C$ then the claim of Theorem~\ref{thm:main-local} holds.
\end{lemma}

\begin{proof}
Since two non-equivalent strings of equal length $l$ cannot have an overlap of length $l$ or greater, it follows that in this case $\gamma \le 1$. Therefore $\beta \le \frac{1}{2}\gamma \le \frac{1}{2}$ and $2\beta+7\gamma \le 8$, a much stronger bound than needed.
\end{proof}

In the remainder of this section we assume that not all $l_i$ are equal.

\begin{lemma}
\label{lem:mainclaim}
Either of the following statements imply the claim of Theorem~\ref{thm:main-local} for a $k$-cycle $C$:
\begin{itemize}
\item $2M - 7\Delta O \le \frac{1}{2}$,
\item $\Delta O \ge \frac{6-k}{2(7k+2)}L.$ 
\end{itemize}
\end{lemma}

Before proving the above lemma, let us note its particularly useful consequences:
\begin{corollary}
\label{cor:mainclaim}
Let $C$ be a $k$-cycle, then the claim of Theorem~\ref{thm:main-local} for $C$
\begin{itemize}
\item is implied by $\Delta O \ge \frac{1}{30}L$ if $k=4$,
\item is implied by $\Delta O \ge \frac{1}{74}L$ if $k=5$,
\item holds if $k \ge 6$.
\end{itemize}
\end{corollary}

\begin{proof}[Proof (of Lemma~\ref{lem:mainclaim})]
For the first part we have
\[ 2M+7O = 2M + 7\left(\frac{3}{2}L-\Delta O \right) = 10\frac{1}{2}L + (2M-7\Delta O),\]
and the claim follows.

For the second part, note that $M \le \frac{1}{k}O$. Therefore if we have  $\Delta O \ge \frac{6-k}{2(7k+2)}L$, then
\[ 2M+7O \le \frac{2+7k}{k}O \le \frac{2+7k}{k}\left(\frac{3}{2}L-\frac{6-k}{2(7k+2)}L\right) =
\frac{2+7k}{k} \cdot \frac{22k}{2(7k+2)} L = 11L.\]
\end{proof}

\subsection{The Proof of Theorem~\ref{thm:main-local} for $5$-cycles}
\label{subsec:5-cycles}

The remainder of the proof is divided into four parts, one for each cycle length in $\{2,3,4,5\}$. Although there are similarities between these parts, they are mostly independent. For easier reading, we put each part in a separate subsection. 

\begin{lemma}
If $C$ is $5$-cycle, then $2M+7O \le 11L$.
\end{lemma}

\begin{proof}
We consider three cases. In all three we prove that $\Delta O \ge \frac{1}{74}L$ and the claim follows from Corollary~\ref{cor:mainclaim}.

\case{1} If $\lmin > \frac{1}{2} \lmax$ we have by Lemma~\ref{lem:delta-flat} that $ \Delta O \ge \frac{1}{4}\lmin$ and so
\[ \Delta O \ge \frac{1}{36}\Big(9\lmin\Big) \ge \frac{1}{36}\Big(\lmin + 4\lmax\Big) \ge \frac{1}{36}L.\]

\case{2} If $\frac{1}{4}\lmax < \lmin \le \frac{1}{2} \lmax$, then we cannot have two consecutive steep down-edges, and so by Lemma~\ref{lem:delta-app} we have
\[ \Delta O \ge \frac{1}{8}\Big(\lmax-\lmin\Big) \ge \frac{1}{72}\Big(9\lmax-9\lmin\Big) \ge \frac{1}{72}\Big(4\lmax+\lmin\Big) \ge \frac{1}{72}L.\]

\case{3} Finally, if $\lmin \le \frac{1}{4} \lmax$, then  we have
\[ \Delta O \ge \frac{1}{12}\Big(\lmax-\lmin\Big) \ge \frac{1}{72}\Big(6\lmax-6\lmin\Big) \ge \frac{1}{72}\Big(4\lmax+2\lmin\Big) \ge \frac{1}{72} L.\] 
\end{proof}

\subsection{The Proof of Theorem~\ref{thm:main-local} for $4$-cycles}
\label{subsec:4-cycles}

\begin{lemma}
\label{lem:4-cycle}
If $C$ is $4$-cycle, then $2M+7O \le 11L$.
\end{lemma}

\begin{proof}
We again consider several cases.

\case{1} If $\lmin > \frac{1}{2} \lmax$ then by Lemma~\ref{lem:delta-flat} we have $\Delta O \ge \frac{1}{4}\lmin$ and so
\[ \Delta O \ge \frac{1}{28}\Big(7\lmin\Big) \ge \frac{1}{28}\Big(\lmin + 3\lmax\Big) \ge \frac{1}{28}L,\]
and the claim follows by Corollary~\ref{cor:mainclaim}.

\case{2} If $l_{\min} \le \frac{1}{2} l_1$, but all down-edges of $C$ are flat then we have by Lemma~\ref{lem:delta-app} that $\Delta O \ge \frac{1}{4}(\lmax - \lmin)$ and so
\[ \Delta O \ge \frac{1}{28}\Big(7\lmax - 7\lmin\Big) \ge \frac{1}{28}\Big(3\lmax+\lmin\Big) \ge \frac{1}{28}L.\]

Therefore we only need to consider cases where at least one down-edge of $C$ is steep. This is implicitly assumed in all remaining cases.

\case{3} If $l_1 \ge l_2 \le l_3 \ge l_4$, i.e.\ the edges of $C$ are alternating down-up-down-up, then
\[ \Delta O = \Big(\Delta o_{41}+\Delta o_{12}\Big)+\Big(\Delta o_{23}+\Delta o_{34}\Big) \ge \frac{1}{6}\Big(l_1-l_2\Big)+\frac{1}{6}\Big(l_3-l_4\Big),\]
by Corollary~\ref{cor:delta}.
We also have \[M \le \min\Big(o_{41},o_{23}\Big) \le \min\Big(l_1,l_3\Big) \le \frac{1}{2}\Big(l_1+l_3\Big).\]
Therefore
\[ 2M-7\Delta O \le l_1+l_3 + \frac{7}{6}\Big(-l_1+l_2-l_3+l_4\Big) =
\frac{1}{6}\Big(-l_1+7l_2-l_3+7l_4\Big) \le \frac{1}{2}\Big(l_1+l_2+l_3+l_4\Big)\]
since $l_2 \le l_1$ and $l_4 \le l_3$.

\case{4} If $l_1 \ge l_2 \le l_3 \le l_4$, then we have
$\Delta O \ge \frac{1}{6}(l_1-l_2)$ and $M \le o_{23} \le l_3$. Therefore
\[ 2M-7\Delta O \le 2l_3-\frac{7}{6}l_1 + \frac{7}{6}l_2 \le 
\Big(\frac{1}{2}l_3+\frac{1}{2}l_4+l_1\Big)-\frac{7}{6}l_1+\Big(\frac{3}{6}l_2 + \frac{4}{6}l_1\Big) = \frac{1}{2}L,\]
and the claim follows from Lemma~\ref{lem:mainclaim}.

\case{5} If $l_1 \ge l_2 \ge l_3 \le l_4$, we consider two subcases. Since we excluded Cases 1 and 2, at least one down-edge of $C$ is steep.

\case{5a} If $l_2 \le \frac{1}{2}l_1$, then $\Delta O \ge \frac{1}{6}(l_1-l_2)$ and $M \le o_{34} \le l_3 + \frac{1}{2}l_4$. Therefore
\[ 2M-7\Delta O \le 2l_3+l_4 + \frac{7}{6}l_2-\frac{7}{6}l_1 \le \Big(\frac{1}{2}l_3 + \frac{3}{4}l_1\Big)+\Big(\frac{1}{2}l_4+\frac{1}{2}l_1\Big)+\Big(\frac{3}{6}l_2+\frac{2}{6}l_1\Big)-\frac{7}{6}l_1 \le \frac{1}{2}L \]

\case{5b} If $l_2 > \frac{1}{2}l_1$ and $l_3 \le \frac{1}{2}l_2$, then $\Delta O \ge \frac{1}{8}(l_1-l_3)$ by Lemma~\ref{lem:delta-app} and $M \le o_{34} \le l_3 + \frac{1}{2}l_4$. Therefore
\[ 2M-7\Delta O \le 2l_3 + l_4 - \frac{7}{8}l_1 + \frac{7}{8}l_3 \le \frac{23}{8}l_3 + l_4 - \frac{7}{8}l_1 \le 
\Big(\frac{4}{8}l_3 + \frac{4}{8}l_2 + \frac{11}{16}l_1\Big)+\Big(\frac{1}{2}l_4 + \frac{1}{2}l_1\Big)-\frac{7}{8}l_1 \le \frac{1}{2}L.\]

\case{6} We are left with the case where $l_1 \ge l_2 \ge l_3 \ge l_4$, i.e.\ $C$ has three down-edges, and at least one of them is steep. We consider three subcases:

\case{6a} If $l_2 \le \frac{1}{2} l_1$ then similarly to Case 3 we have
\[\Delta O \ge \Big(\Delta o_{41}+\Delta o_{12}\Big)+\Big(\Delta o_{23}+\Delta o_{34}\Big) \ge \frac{1}{6}\Big(l_1-l_2\Big)+\frac{1}{6}\Big(l_3-l_4\Big).\]
We also have $M \le o_{34} \le l_3 + \frac{1}{2}l_4$. Therefore
\begin{align*}
2M-7\Delta O \le 2l_3+l_4 +\frac{7}{6}\Big(-l_1+l_2-l_3+l_4\Big)= -\frac{7}{6}l_1+\frac{7}{6}l_2+\frac{5}{6}l_3+\frac{13}{6}l_4 \le\\
\le-\frac{7}{6}l_1+\Big(\frac{3}{6}l_2+\frac{2}{6}l_1\Big)+\Big(\frac{3}{6}l_3+\frac{1}{6}l_1\Big)+\Big(\frac{3}{6}l_4+\frac{5}{6}l_1\Big) \le \frac{1}{2}L.
\end{align*}

\case{6b} If $l_3 \le \frac{1}{2} l_2$ then $\Delta O \ge \frac{1}{6}(l_2 - l_3)$ and $M \le l_3 + \frac{1}{2}l_4$. Therefore
\begin{align*}
2M - 7\Delta O \le 2l_3 + l_4 - \frac{7}{6}l_2 + \frac{7}{6}l_3 = -\frac{7}{6}l_2 + \frac{19}{6}l_3 + l_4 \le -\frac{7}{6}l_2 + \frac{22}{6}l_3 + \frac{1}{2}l_4 \le\\
\le -\frac{7}{6}l_2+\Big(\frac{3}{6}l_1+\frac{13}{12}l_2+\frac{3}{6}l_3\Big)+\frac{1}{2}l_4 \le \frac{1}{2}L.
\end{align*}

\case{6c} If $l_4 \le \frac{1}{2} l_3$ then similarly to Case 6a we have
\[\Delta O \ge \frac{1}{6}\Big(l_1-l_2\Big)+\frac{1}{6}\Big(l_3-l_4\Big),\]
but this time we use the bound $M \le o_{41} \le l_4 + \frac{1}{2}l_1$. We get
\begin{align*}
2M-7\Delta O \le 2l_4 + l_1 + \frac{7}{6}\Big(-l_1+l_2-l_3+l_4\Big) = -\frac{1}{6}l_1 + \frac{7}{6}l_2 - \frac{7}{6}l_3  + \frac{19}{6}l_4 \le \\
\le -\frac{1}{6}l_1 + \Big(\frac{3}{6}l_2 + \frac{4}{6}l_1\Big)  -\frac{7}{6}l_3 + \Big(\frac{3}{6}l_4 + \frac{8}{6}l_3\Big) \le \frac{1}{2}L.
\end{align*}

\end{proof}

%%%%%%%%%%%%%%%%%%%%%%%%%%%%%%%%%%%%%%%%%%%%%%%%%%%%%%%%%%%%%%%%%%%%%%%%%%%%%%%%%%

\subsection{The Proof of Theorem~\ref{thm:main-local} for $3$-cycles}
\label{subsec:3-cycles}

\begin{lemma}
If $C$ is a $3$-cycle then $2M+7O \le 11L$.
\end{lemma}

\begin{proof}
There are essentially two kinds of $3$-cycles - ones with (cyclically) increasing $l_i$, and ones with decreasing $l_i$. 

\case{1} If $l_1 \ge l_2 \le l_3$ (i.e.\ $l_i$ are cyclically increasing), then we consider three subcases.

\case{1a} If $\lmax < 2\lmin$, i.e.\ $l_1 < 2l_2$, then
\[ \Delta O \ge \Delta o_{23} + \Delta o_{31} \ge \Big(l_2 - \frac{1}{2}l_3\Big) + \Big(l_3 - \frac{1}{2}l_1\Big) = l_2+\frac{1}{2}l_3 - \frac{1}{2}l_1 > \frac{1}{2}l_3\]
by Lemma~\ref{lem:delta-up}. We also have $M \le o_{23} \le l_3$. Then
\[ 2M - 7\Delta O \le 2l_3 - \frac{7}{2}l_3 < 0 \le \frac{1}{2}L.\]

\case{1b} If $l_1 \ge 2l_3$ then we have 
\[\Delta O \ge \Delta o_{12} \ge \frac{1}{6}\Big(l_1-l_2\Big)\]
and $M \le o_{23} \le l_2 + \frac{1}{2}l_3$. Therefore
\[ 2M-7\Delta O \le 2l_2 + l_3 +\frac{7}{6}\Big(-l_1+l_2\Big) \le 
-\frac{7}{6}l_1 + \frac{19}{6}l_2+l_3 \le 
-\frac{7}{6}l_1 + \Big(\frac{3}{6}l_2 + \frac{8}{6}l_1\Big) + \Big(\frac{1}{2}l_3 + \frac{1}{4}l_1\Big) \le \frac{1}{2}L.\]

\case{1c} If $2l_2 \le l_1 < 2l_3$ then
\[ \Delta O \ge \max\Big(\Delta o_{12}, \Delta o_{31}\Big) \ge \max\Big(\frac{1}{6}\Big(l_1-l_2\Big),l_3-\frac{1}{2}l_1\Big)\]
and $M \le o_{23} \le l_2 + \frac{1}{2}l_3$. Hence
\[ 2M-7\Delta O \le 2l_2 + l_3 - \Big(l_3-\frac{1}{2}l_1\Big)-6\cdot\frac{1}{6}\Big(l_1-l_2\Big) = -\frac{1}{2}l_1 + 3l_2 \le -\frac{1}{2}l_1 + \Big(l_1+\frac{1}{2}l_2+\frac{1}{2}l_3\Big)\le \frac{1}{2}L.\] 

\case{2} If $l_1 \ge l_2 \ge l_3$ then we consider several subcases. The logic in their ordering is that we are trying to eliminate the easy ones first until only the hardest case remains --- one that is actually tight.

\case{2a} If $l_1 < 2l_3$, i.e.\ $\lmax < 2\lmin$, then by Lemma~\ref{lem:delta-flat} we have $\Delta O \ge \frac{1}{4} l_3$. We also have $M \le \min(o_{23},o_{31}) \le \min(l_2+\frac{1}{2}l_3,l_1)$. Therefore
\[ 2M-7\Delta O \le \frac{1}{2}\Big(l_2+\frac{1}{2}l_3\Big)+\frac{3}{2}l_1 - \frac{7}{4}l_3 = \frac{3}{2}l_1 + \frac{1}{2}l_2 - \frac{3}{2}l_3 \le \frac{1}{2}L.\]

\case{2b} If $l_1 \ge 2l_3$, but both down-edges are flat, then by Lemma~\ref{lem:delta-app} we have $\Delta O \ge \frac{1}{4}(l_1-l_3)$. Using the same bound on $M$ as in the previous case, we obtain
\[ 2M-7\Delta O \le \frac{1}{2}\Big(l_2+\frac{1}{2}l_3\Big)+\frac{3}{2}l_1 - \frac{7}{4}\Big(l_1-l_3\Big) = -\frac{1}{4}l_1 + \frac{1}{2}l_2 + 2l_3 \le \frac{1}{2}L.\]

We are left with the case where at least one down-edge of $C$ is steep.

\case{2c} If $l_2 \le \frac{1}{2}l_1$ then we have $\Delta O \ge \frac{1}{6}(l_1-l_2)$ by Corollary~\ref{cor:delta}. We also have $M \le o_{23} \le l_2 + \frac{1}{2}l_3$. Hence
\[ 2M-7\Delta O \le 2l_2 + l_3 - \frac{7}{6}l_1 + \frac{7}{6}l_2 = -\frac{7}{6}l_1 + \frac{19}{6}l_2 + l_3 \le 
-\frac{7}{6}l_1 + \Big(\frac{3}{6}l_2 + \frac{8}{6}l_1\Big) + \Big(\frac{1}{2}l_3 + \frac{1}{4}l_1\Big) \le \frac{1}{2}L. \]

\case{2d} If $l_2 > \frac{1}{2} l_1$ and $2l_3 \le l_2 < \frac{5}{2}l_3$ then we have
\[ \Delta O \ge \Delta o_{12} + \Delta o_{23} \ge \Delta \alpha_2 + \Delta o_{23} \ge \Big(\frac{3}{2}l_2 + \frac{1}{2}l_3\Big) - \Big(2l_2 - l_3\Big) = \frac{3}{2}l_3 - \frac{1}{2}l_2\]
by Lemma~\ref{lem:delta}, and we also have
\[ \Delta O \ge \Delta o_{12} \ge \frac{1}{2}\Big(l_1-l_2\Big). \]
Joining these two bounds with $M \le o_{31} \le l_3 + \frac{1}{2}l_1$ gives

\[ 2M - 7\Delta O \le 2l_3 + l_1 - 6\Big(\frac{3}{2}l_3-\frac{1}{2}l_2\Big) -\frac{1}{2}\Big(l_1-l_2\Big) =
\frac{1}{2}l_1+\frac{7}{2}l_2-7l_3 \le \frac{1}{2}l_1 + \Big(\frac{1}{2}l_2 + \frac{15}{2}l_3\Big) - 7l_3 = \frac{1}{2}L.
\]

\case{2e} Finally if $l_2 > \frac{1}{2} l_1$ and $l_2 \ge \frac{5}{2}l_3$ then we have
\[ \Delta O \ge \Delta o_{12} + \Delta o_{23} \ge \Delta \alpha_2 + \Delta o_{23} \ge \Big(\frac{3}{2}l_2 + \frac{1}{2}l_3\Big) - \Big(l_2 + l_3 + \alpha_3\Big) \ge \frac{1}{2}l_2 - l_3\]
by Lemma~\ref{lem:delta}. We now proceed similarly to the previous case:
\[ 2M - 7 \Delta O \le 2l_3 + l_1 - 6\Big(\frac{1}{2}l_2-l_3\Big)-\frac{1}{2}\Big(l_1-l_2\Big) \le 
\frac{1}{2}l_1 - \frac{5}{2}l_2 + 8l_3 = \frac{1}{2}L -3l_2+ \frac{15}{2}l_3 \le \frac{1}{2}L .\] 

\end{proof}

%%%%%%%%%%%%%%%%%%%%%%%%%%%%%%%%%%%%%%%%%%%%%%%%%%%%%%%%%%%%%%%%%%%%%%%%%%%%%%%%

\subsection{The Proof of Theorem~\ref{thm:main-local} for $2$-cycles}
\label{subsec:2-cycles}

Before we proceed with the case of $2$-cycles, we need an additional technical lemma.

\begin{lemma}
\label{lem:delta-2-cycle}
If $l_1 \ge 2l_2$ then $\Delta o_{12} + \Delta o_{21} \ge \frac{1}{2}l_2$. 
\end{lemma}

\begin{proof}
If $o_{12} \le l_1$ then clearly $\Delta O \ge \Delta o_{12} \ge \frac{1}{2}l_2$. Hence we can assume $o_{12} \ge l_1$. Note that this means that $l_1$ is not a multiple of $l_2$, since $w_1$ is primitive.

Assume w.l.o.g.\ that $w_2$ is its maximal rotation and let $k \ge 2$ be such, that $ kl_2 < l_1 < (k+1)l_2$. Since $o_{12} \ge l_1$, by Lemma~\ref{lem:general-long-l1} we get $\imin(w_{12}) = \alpha_2+1$ and $\imax(w_{12}) \ge kl_2+1$. This means that $|\pmin(w_1)| \ge kl_2-\alpha_2 = (k-1)l_2 + (l_2-\alpha_2) \ge \frac{3}{2}l_2$ and $|\pmax(w_1)| \le l_1 - kl_2 + \alpha_2 < \frac{3}{2}l_2$. Therefore, $w_1$ is its maximal rotation as well.

Since $\imax(w_{12}) \ge kl_2+1$ and $l_2$ does not divide $l_1$, we know $\pmax(w_1) = w \pmax(w_2)$, where $w = w_{12}[\imax(w_{12}),l_1]$. Note that $|w| < |w_2|$ and $w$ is a prefix of $w_2$ (because it is an initial segment of a maximal rotation of $w_1$). 

We will show that $o_{21} < |\pmax(w_1)| = \alpha_1$, which implies the claim of the lemma. Assume the opposite, i.e.\ $o_{21} \ge \alpha_1$. Then $o_{21}[1,\alpha_1] = \pmax(w_1) = w \pmax(w_2)$. By Lemma~\ref{lem:unique} this can only happen if $w_2$ is aligned with position $|w|+1$ of $\ov_{21}$. But then $w$ is a suffix of $w_2$, and since it is also a prefix of $w_2$ we get a contradiction with Lemma~\ref{lem:unbordered}.
\end{proof}

\begin{lemma}
If $C$ is a $2$-cycle, then $2M+7O \le 11L$.
\end{lemma}

\begin{proof}
We consider three cases.

\case{1} If $l_1 < 2l_2$, i.e.\ the down-edge of $C$ is flat, then we have
\[\Delta O \ge \max\Big(\frac{1}{2}\Big(l_1-l_2\Big),l_2-\frac{1}{2}l_1\Big) \]
by Corollary~\ref{cor:delta} and Lemma~\ref{lem:delta-up}. We also have $M \le o_{21} \le l_1$, and so
\[ 2M-7\Delta O \le 2l_1 - 5\cdot \frac{1}{2}\Big(l_1-l_2\Big)-2\Big(l_2-\frac{1}{2}l_1\Big) = \frac{1}{2}L.\]

\case{2} If $2l_2 \le l_1 < 3l_2$ then we have $\Delta O \ge \frac{1}{2}l_2$ by Lemma~\ref{lem:delta-2-cycle}. This easily gives our main claim since using $M \le o_{21} \le l_2 + \frac{1}{2}l_1$ we have
\[ 2M-7\Delta O \le 2l_2 + l_1 - \frac{7}{2}l_2 = l_1 - \frac{3}{2}l_2 \le \Big(\frac{1}{2}l_1 + \frac{3}{2}l_2\Big) - \frac{3}{2}l_2 \le \frac{1}{2}L.\]

\case{3} If $3l_2 \le l_1$ then by Corollary~\ref{cor:delta} we have  $\Delta O \ge \frac{1}{4}(l_1-l_2)$ and together with $M \le o_{21} \le l_2 + \frac{1}{2}l_1$ we get
\[ 2M-7\Delta O \le 2l_2 + l_1 - \frac{7}{4}l_1 + \frac{7}{4}l_2 = -\frac{3}{4}l_1 + \frac{15}{4}l_2 = -\frac{3}{4}l_1 + \frac{5}{4}l_1 \le \frac{1}{2}L.\]

\end{proof}

\section{Tight examples}
\label{sec:tight}

\subsection{Tightness of Theorem~\ref{thm:main-local}}

We will now show that Theorem~\ref{thm:main-local} is essentially tight. To this end, we give two examples of cycles in the overlap graph, for which $2M + 7O = 11L - O(1)$. Note that by increasing the lengths of the strings in these cycles we can get $\frac{2M+7O}{11L} \rightarrow 1$.

\begin{example}
\label{ex:tight2}
Let $w_1=ba^kb|a^{k+1}ba^{k+1}$ and $w_2=a^{k+1}|ba^kb$ (we use the symbol $|$ to mark the border between $\pmax$ and $\pmin$). Here $l_1 = 3k+5$, $l_2 = 2k+3$, so $L= 5k+8$. 

Now, let $x_i = (w_i^2)[1,2l_i-1]$ for $i=1,2$. Note that all $w_i$ are nice words and every $x_i$ is a $w_i$-word. 

We have $o_{12} = 4k+5$ and $o_{21} = 3k+4$, so $O = 7k+9$ and $M=3k+4$. Note that $2M+7O = 6k+8+49k+63=55k+71 = 11L-O(1)$
\end{example}

\begin{example}
\label{ex:tight3}
Let $w_1 = ba^nba^{n+1}ba^nb|a^{n+1}ba^{n+1}ba^{n+1}$, $w_2 = a^{n+1}ba^{n+1}|ba^nba^{n+1}ba^nb$, $w_3 = a^{n+1}|ba^nb$. We have $l_1 = 6n+10$, $l_2 = 5n+8$, $l_3 = 2n+3$, so $L=13n+21$. 

Now, let $x_1 = (w_1^2)[1,2l_1-1]$, $x_2 = (w_2^3)[1,2l_2+\alpha_2-1]$ and $x_3 = (w_3^4)[1,4l_3-1]$. Note that all $w_i$ are nice words and every $x_i$ is a $w_i$-word.

We have $o_{12} = 8n+12$, $o_{23} = 6n+8$ and $o_{31} = 5n+7$, so $O=19n+27$ and $M=5n+7$. Note that $2M+7O = 10n+14+133n+189=11L-O(1)$.
\end{example}

We will now show that the bound we give on $\min(M,(1-c)O)$ in Corollary~\ref{cor:factor} is also essentially tight.

Consider a cycle cover $\mathcal{C}$ in the overlap graph, composed of two collections of cycles: $\mathcal{C}_2$ consisting of $2$-cycles of the form described in Example~\ref{ex:tight2}, and $\mathcal{C}_3$ consisting of $3$-cycles described in Example~\ref{ex:tight3} (note that these cycles need to use different $n$ so that their vertices are non-equivalent).

Let $L_2,L_3$ be the total length of the periods of the strings in the cycles of $\mathcal{C}_2$ and $\mathcal{C}_3$, respectively. Let $O_2$ be the sum of all overlaps on the cycles in $\mathcal{C}_2$ and let $M_2$ be the sum of smallest overlaps for each cycle in $\mathcal{C}_2$. Similarly define $O_3$ and $M_3$ for $\mathcal{C}_3$. Finally let $L=L_2+L_3$, $O=O_2+O_3$ and $M=M_2+M_3$. 

Note that to make the analysis in Corollary~\ref{cor:factor} tight we only need to make $M=(1-c)O$, since we already have $2M+7O = 11L-O(1)$.
Since $M_2 \sim \frac{3}{7}O_2$ and $M_3 \sim \frac{5}{19}O_3$, this can be done by adjusting the balance between $L_2$ and $L_3$, provided that $c \in [\frac{4}{7},\frac{14}{19}]$. The current best approximation ratio of $\frac{2}{3}$ for \tsp\ sits well within this interval.

\subsection{The Greedy Algorithm}

Recall the greedy algorithm, which picks two strings with the largest overlap and combines them together until a single string remains. The bounds in Breslauer et al.\ can be used to improve the analysis of this algorithm, as shown by Kaplan et al.~\cite{kaplan-greedy}. It is natural to ask whether our bounds can be used in a similar fashion. Unfortunately, it seems that there is no simple way to do this. In their analysis, Kaplan et al.\ require a good bound on the the total overlap of a (possibly) long path of strings in the overlap graph. As it turns out, in this case the overlap can actually approach the bound of $\frac{3}{2}\sum_i l_i$ arbitrarily close, as can be seen in the following example.

\begin{example}
For any $k \ge 1$ let $w_{2k} = b^k|a^k$ and $w_{2k-1} = a^{k-1}|b^k$.
Also, let $x_{2k} = b^k a^k b^k a^{k-1}$ and $x_{2k-1} = a^{k-1} b^k a^{k-1} b^{k-1}$. Note that all $w_i$ are nice words and every $x_i$ is a $w_i$-word.

Consider $S = \{x_3,x_4,\ldots,x_n\} = \{ab^2ab, b^2a^2b^2a, a^2b^3a^2b^2,b^3a^3b^3a^2,\ldots\}$ and the path $x_n \rightarrow x_{n-1} \rightarrow \ldots \rightarrow x_3$ in the overlap graph of $S$. It is easy to verify that $o_{i+1,i} = \lfloor \frac{3i}{2} \rfloor$. Therefore, the total overlap of the path is approximately $\frac{3}{4}n^2$, and $\sum_{i=3}^n l_i \sim \frac{1}{2}n^2$.
\end{example}

This, of course, does not rule out using our results to improve the analysis of the greedy algorithm. However, any such result requires some additional insight.

\section{Acknowledgements}

The author would like to thank Anupam Gupta and Piotr Sankowski for valuable suggestions that helped to improve the presentation of the paper, and Jakub Radoszewski for his advice in matters of stringology. 

%%%%%%%%%%%%%%%%%%%%%%%%%%%%%%%%%%%%%%%%%%%%%%%%%%%%%%%%%%%%%%%%%%%%%%%%%%%%%%%%%%%%%%%%%%

\bibliographystyle{plain}
\bibliography{superstring}

%%%%%%%%%%%%%%%%%%%%%%%%%%%%%%%%%%%%%%%%%%%%%%%%%%%%%%%%%%%%%%%%%%%%%%%%%%%%%%%%%%%%%%%%%%

\end{document}